\documentclass{amsart}

\newtheorem{thm}{Theorem}

\title[]{Hamiltonians representing equations of motion with damping due to friction}

\author{Stephen Montgomery-Smith}
\address{Department of Mathematics, University of Missouri, Columbia, MO 65211, U.S.A.}


\begin{document}

\begin{abstract}
Suppose that $H(q,p)$ is a Hamiltonian on a manifold $M$, and $\tilde L(q,\dot q)$, the Rayleigh dissipation function, satisfies the same hypotheses as a Lagrangian on the manifold $M$.  We provide a Hamiltonian framework that gives the equation
\begin{equation*}
\dot q = \frac{\partial H}{\partial p}(q,p) , \quad
\dot p = - \frac{\partial H}{\partial q}(q,p) - \frac{\partial \tilde L}{\partial \dot q}(q,\dot q)
\end{equation*}
The method is to embed $M$ into a larger framework where the motion drives a wave equation on the negative half line, where the energy in the wave represents heat being carried away from the motion.  We obtain a version of N\"other's Theorem that is valid for dissipative systems.  We also show that this framework fits the widely held view of how Hamiltonian dynamics can lead to the ``arrow of time.''
\end{abstract}

\maketitle

\section{Introduction}

The purpose of this document is to provide a Hamiltonian framework which gives rise to equations of motion that include friction or damping effects.  The author has found a number of other works \cite{bruneau,chandrasekar,froelich,herrera,koopman,rabei,riewe1,riewe2,smith,yu,zmoginov}, but none of them seem to solve the problem in the manner presented here.  The methods of this paper allow for extremely general formulas for the damping term, but the method is also artificial.  Nevertheless we are able to extract various principles such as N\"other's Theorem.

\section{Notation}

We represent the position of a particle by $q$, which might be in a finite or infinite dimensional manifold $M$.  (Everything in this document is formal, and there is no attempt at rigor.)  As usual we have a Lagrangian $L(q,\dot q)$ which is strictly convex in the second coordinate.  The equations of motion are the solution to the variational equation $\delta \mathcal S = 0$, where the action $\mathcal S$ is
\begin{equation}
\mathcal S = \int_{T_0}^{T_1} L(q,\dot q) \, dt
\end{equation}
that is
\begin{equation}
\frac d{dt}\left(\frac{\partial L}{\partial \dot q}\right) - \frac{\partial L}{\partial q} = 0
\end{equation}
As usual, we create the Hamiltonian
\begin{equation}
\label{H from L}
H(q,p) = p \cdot \dot q - L(q,\dot q)
\end{equation}
where $p$ satisfies the equation
\begin{equation}
\label{p from dot q}
p = \frac{\partial L}{\partial \dot q}
\end{equation}
It can then be shown that
\begin{gather}
\label{dot q from p}
\dot q = \frac{\partial H}{\partial p} \\
\label{Lq = -Hq}
\frac{\partial L}{\partial q} = - \frac{\partial H}{\partial q}
\end{gather}
where in the latter formula one has to be aware that the first partial derivative is keeping $\dot q$ constant, and the second partial derivative is keeping $p$ constant.  The equation of motion becomes
\begin{equation}
\dot p = - \frac{\partial H}{\partial q}
\end{equation}
The usual example is
\begin{equation}
L(q,\dot q) = \tfrac12 \, \dot q \cdot A \cdot \dot q - V(q)
\end{equation}
where $A = A(q)$ is positive definite symmetric.  Then
\begin{gather}
p = A \cdot \dot q \\
H(q,p) = \tfrac12 \, p \cdot A^{-1} \cdot p + V(q)
\end{gather}
and the equation of motion is
\begin{equation}
\frac d{dt} ( A \cdot \dot q) - \tfrac12 \, \dot q \cdot \frac{\partial A}{\partial q} \cdot \dot q + \frac{\partial V}{\partial q} = 0
\end{equation}
The book \cite{arnold} gives many more details.

\section{Striving to understand friction}
\label{striving}

Here we describe a widely held view of how friction arises from reversible Hamiltonian dynamics.  Our concrete equations of motion obtained in Sections~\ref{proof of concept} and~\ref{wave equation} should conform to the framework described here.

Suppose the manifold $M$ is a submanifold of $M_0 \times M_1$, where where $M_0$ represents the macroscopic motion, and $M_1$ represents the microscopic motion.  The position then splits as $q = (q_0,q_1)$, and the momentum splits as $p = p_0 + p_1$ in the obvious manner.  Because $M_1$ represents microscopic motion, it makes sense that $M_1$ should have much higher dimension than $M_0$, indeed we will often suppose that $M_1$ is infinite dimensional.  For example, maybe the microscopic motion is a large number of high frequency harmonic oscillators.

We assume that the Hamiltonian splits as
\begin{equation}
H(q,p) = H_0(q_0,p_0) + H_1(q_0,q_1,p_1)
\end{equation}
Here $H_0$ represents the macroscopic motion without friction.  The second part, $H_1$, is the energy of the microscopic motion plus coupling between the microscopic and macroscopic parts.  We assume that all the interaction between the microscopic and macroscopic parts is via position and not momentum.

Because $H_1$ does not depend upon $p_0$, from equation~\eqref{dot q from p} we obtain
\begin{gather}
\dot q_0 = \frac{\partial H_0}{\partial p_0} \\
\dot q_1 = \frac{\partial H_1}{\partial p_1}
\end{gather}
and the Lagrangian splits as
\begin{equation}
L(q,\dot q) = L_0(q_0,\dot q_0) + L_1(q_0,q_1,\dot q_1)
\end{equation}
where
\begin{gather}
L_0(q_0,\dot q_0) = p_0 \cdot \dot q_0 - H_0(q_0,p_0) \\
L_1(q_0,q_1,\dot q_1) = p_1 \cdot \dot q_1 - H_1(q_0,q_1,p_0)
\end{gather}
and equations~\eqref{p from dot q} and~\eqref{Lq = -Hq} become
\begin{gather}
p_0 = \frac{\partial L_0}{\partial \dot q_0} \\
p_1 = \frac{\partial L_1}{\partial \dot q_1} \\
\label{L0q = -H0q}
\frac{\partial L_0}{\partial q} = - \frac{\partial H_0}{\partial q} \\
\label{L1q = -H1q}
\frac{\partial L_1}{\partial q} = - \frac{\partial H_1}{\partial q}
\end{gather}
(The only mildly non-obvious statements are equation~\eqref{L0q = -H0q} which follows from ${\partial L_0}/{\partial q_0} = - {\partial H_0}/{\partial q_0}$ and ${\partial L_0}/{\partial q_1} = {\partial H_0}/{\partial q_1} = 0$, and equation~\eqref{L1q = -H1q} which follows by subtraction.)

The equations of motion now become
\begin{gather}
\dot p_0 = - \frac{\partial H_0}{\partial q_0} - \frac{\partial H_1}{\partial q_0} \\
\dot p_1 = \phantom{- \frac{\partial H_0}{\partial q_1}} - \frac{\partial H_1}{\partial q_1} 
\end{gather}
Now, if the microscopic motion starts at rest, then we \emph{hope} that the second law of thermodynamics should be obeyed, and as a function of time, $H_0$ should be non-increasing, and $H_1$ should be non-decreasing.   Thus the hope is that the motion described by $q_1$ somehow provides a memory of what has happened to $q_0$, so that we end up with
\begin{equation}
\frac{\partial H_1}{\partial q_0} = \Phi(q_0,\dot q_0)
\end{equation}
for some function $\Phi$.
Then $H_0$ satisfies
\begin{equation}
\begin{aligned}
\dot H_0 &= \dot q_0 \cdot \frac{\partial H_0}{\partial q_0} + \dot p_0 \cdot \frac{\partial H_0}{\partial p_0} \\& = \dot q_0 \cdot (-\dot p_0 - \Phi(q_0,\dot q_0)) + \dot p_0 \cdot \dot q_0 = - \dot q_0 \cdot \Phi(q_0, \dot q_0)
\end{aligned}
\end{equation}
If $\Phi$ satisfies a property like
\begin{equation}
\label{Phi prop}
y \cdot \Phi(x,y) \ge 0
\end{equation}
then we that $H_0$ is non-increasing in time.

The case we can handle with our approach is to consider is
\begin{equation}
\Phi(x,y) = \frac{\partial \tilde L}{\partial y}(x,y)
\end{equation}
where $y\mapsto \tilde L(x,y)$ is a strictly convex.  Note that property~\eqref{Phi prop} holds if $\tilde L(x,y)$ attains its minimum at $y = 0$, because
\begin{equation}
\begin{aligned}
y \cdot \frac{\partial \tilde L}{\partial y}(x,y)
&= y \cdot \frac{\partial \tilde L}{\partial y}(x,0) + \int_0^1 \frac\partial{\partial \lambda} \left( y \cdot \frac{\partial \tilde L}{\partial y}(x,\lambda y) \right) \, d\lambda \\
&= \int_0^1 y \cdot \frac{\partial^2 \tilde L}{\partial y^2}(x,\lambda y) \cdot y \, d\lambda \ge 0
\end{aligned}
\end{equation}
since the Hessian of $\tilde L(x,\cdot)$ is positive definite.

The initial values of $q_1$ and $p_1$ should be very important, perhaps representing that the parts in $M_1$ are at rest, or are oscillating in such a manner that the oscillations can only increase.  This is because Hamiltonian equations are reversible in time.  But dissipative systems should only be solvable forwards in time.  The is sometimes known as the ``arrow of time.''  For example, we know that if a wine glass is thrown to the ground, then it shatters.  The reverse process, where the wine glass comes back together, is possible in theory if the initial positions and momenta of the molecules and atoms are given in a very precise manner (and we neglect the effects of quantum physics).  But finding this initial data should be very difficult.

\section{Friction via the wave equation: proof of concept}
\label{proof of concept}

We will show how to artificially create damping via the wave equation.  As a proof of concept, let us explain how to create a Hamiltonian such that
\begin{gather}
\label{q f m B}
\dot q_0 = \frac{\partial H_0}{\partial p_0} \\
\label{p f m B}
\dot p_0 = - \frac{\partial H_0}{\partial q_0} - B \cdot \dot q_0
\end{gather}
for some positive definite matrix $B$.

An example of an infinite dimensional Hamiltonian is the wave equation on $C^1(\mathbb R,M)$, whose Hamiltonian is given by
\begin{equation}
\int_{-\infty}^\infty \tfrac12\,p(s)\cdot B^{-1}\cdot p(s) + \tfrac12 \, q'(s)\cdot B \cdot q'(s) \, ds
\end{equation}
Here $q'$ denotes ${\partial q}/{\partial s}$.  Now the equation of motion coming from this Hamiltonian is
\begin{equation}
B \ddot q = B q''
\end{equation}
which by d'Alembert's principle has a general solution
\begin{equation}
\label{d'Alemb}
q(s,t) = \phi_1(s+t) + \phi_2(s-t)
\end{equation}
This is a good example of a Hamiltonian for which Poincar\'e's recurrence theorem does not apply.  And furthermore, the equations of motion are a point of variation of the Lagrangian, but are clearly not a minimum of the Lagrangian:
\begin{equation}
\int_{-\infty}^\infty \tfrac12\,\dot q(s)\cdot B^\cdot \dot q(s) - \tfrac12 \, q'(s)\cdot B \cdot q'(s) \, ds
\end{equation}
The way we simulate friction is to consider the macroscopic motion as driving a wave equation on the half line.  Let $M_1 = C^1((-\infty,0],M_0)$, that is, the set of functions $q:(-\infty,0]\to M_0$ such that $q'$ is bounded, $q'(s) \to q'(0)$ as $s\to 0^-$, where $q'(0)$ denotes the left derivative of $q(s)$ at $s=0$, and $q'(s) \to 0$ as $s \to -\infty$.

Then let
\begin{equation}
M = \{ (q_0, q_1) \in M_0 \times M_1 : q_1(0) = q_0 \}
\end{equation}
with the Hamiltonian
\begin{equation}
H(q,p) = H_0(q_0,p_0) + \int_{-\infty}^0 \tfrac12\,p_1(s)\cdot B^{-1}\cdot p_1(s) + \tfrac12 \, q'_1(s)\cdot B \cdot q'_1(s) \, ds
\end{equation}
We see that the equation of motion is
\begin{gather}
\label{dot q0}
\dot q_0 = \frac{\partial H_0}{\partial p_0} \\
\dot p_0 = - \frac{\partial H_0}{\partial q_0} - B \cdot q'_1(0) \\
\label{dot q1}
B \ddot q_1 = B q_1''
\end{gather}
For example, to obtain equations~\eqref{dot q0} and~\eqref{dot q1}, see that for any infinitesimal perturbation $\delta q = (\delta q_0, \delta q_1)$ of $q = (q_0,q_1)$, noting that $\delta q_0 = \delta q_1(0)$ and $\delta q'_1(s) \to 0$ as $s \to -\infty$, we have
\begin{align*}
\delta q \cdot \frac{\partial H}{\partial q}
&=
\delta q_0 \cdot \frac{\partial H_0}{\partial q_0}(q_0,p_0) + \int_{-\infty}^0 \delta q'_1(s)\cdot B \cdot q'_1(s) \, ds \\
&=
\delta q_0 \cdot \frac{\partial H_0}{\partial q_0}(q_0,p_0) + \Bigl[ \delta q_1(s)\cdot B \cdot q'_1(s) \Bigr]_{-\infty}^0 - \int_{-\infty}^0 \delta q_1(s)\cdot B \cdot q''_1(s) \, ds \\
&=
\delta q_0 \cdot \frac{\partial H_0}{\partial q_0}(q_0,p_0) + \delta q_0\cdot B \cdot q'_1(0) - \int_{-\infty}^0 \delta q_1(s)\cdot B \cdot q''_1(s) \, ds
\end{align*}
Next we impose initial conditions on $(q_1,p_1)$
\begin{equation}
\label{wave init}
q_1(s,T_0) = q_0(T_0), \qquad
p_1(s,T_0) = 0 \qquad (s \le 0)
\end{equation}
that is, the microscopic part of the motion is initially at rest.  Then it can be seen that the solution to the wave equation is given by~\eqref{d'Alemb} with $\phi_1(t) = q_0(t)$ and $\phi_2 = 0$, that is
\begin{equation}
\label{wave sol}
q_1(s,t) = \begin{cases}q_1(s+t-T_0,T_0) = q_0(T_0) & \text{if $s+t \le T_0$} \\ q_0(s+t) & \text{if $s+t \ge T_0$} \end{cases}
\end{equation}
In particular, we see that $q'_1(0,t) = \dot q_1(0,t) = \dot q_0(t)$, and hence we obtain equations~\eqref{q f m B} and~\eqref{p f m B}.

Thus it is seen that the irreversibility of equations of motion with damping comes naturally from the special nature of the initial conditions.  (And indeed any initial condition satisfying $q_1'(s,T_0) = \dot q_1(s,T_1)$ and $q_1(0,T_0) = q_0(T_0)$ will work just as well.)

\section{The general equation of motion with damping}
\label{wave equation}

Suppose that $\tilde H$, and its corresponding Lagrangian $\tilde L$, have the same domains as $H_0$ and $L_0$.  Here $\tilde L$ is also called the Rayleigh dissipation function.

Consider the Hamiltonian on $M = \{ (q_0, q_1) \in M_0 \times M_1 : q_1(0) = q_0 \}$ given by
\begin{equation}
\label{The H}
H(q,p) = H_0(q_0,p_0) + \int_{-\infty}^0 \tilde H(q_1(s),p_1(s)) + \tilde L(q_1(s),q_1'(s)) \, ds
\end{equation}
Note that the corresponding Lagrangian is
\begin{equation}
L(q,\dot q) = L_0(q_0,\dot q_0) + \int_{-\infty}^0 \tilde L(q_1(s),\dot q_1(s)) - \tilde L(q_1(s),q_1'(s)) \, ds
\end{equation}

\begin{thm} The equations of motion for the Hamiltonian given by~\eqref{The H} with initial conditions~\eqref{wave init} imply
\begin{gather}
\dot q_0 = \frac{\partial H_0}{\partial p_0} \\
\dot p_0 = - \frac{\partial H_0}{\partial q_0}(q_0,p_0) - \frac{\partial \tilde L}{\partial \dot q_0}(q_0,\dot q_0)
\end{gather}
\end{thm}

\begin{proof}
First, looking at the equations of motion for $(q_1,p_1)$ we obtain
\begin{gather}
\label{dot ps}
\dot p_1 = - \frac{\partial \tilde H}{\partial q_1}(q_1,p_1) - \frac{\partial \tilde L}{\partial q_1}(q_1,q'_1) + \frac \partial{\partial s}\left(\frac{\partial \tilde L}{\partial q'_1}(q_1,q'_1)\right) \\
\label{dot qs}
\dot q_1 = \frac{\partial \tilde H}{\partial p_1}(q_1,p_1)
\end{gather}
where the last term of equation~\eqref{dot ps} comes by integrating by parts, just as in the previous section.

From the definition of $\tilde L$, we see that equation~\eqref{dot qs} is equivalent to
\begin{equation}
p_1 = \frac{\partial \tilde L}{\partial \dot q_1}(q_1,\dot q_1)
\end{equation}
and substituting into equation~\eqref{dot ps} we obtain
\begin{equation}
\label{wave}
\frac \partial{\partial t}\left(\frac{\partial \tilde L}{\partial \dot q_1}(q_1,\dot q_1)\right) - \frac{\partial \tilde L}{\partial q_1}(q_1,\dot q_1) = \frac \partial{\partial s}\left(\frac{\partial \tilde L}{\partial q'_1}(q_1,q'_1)\right) - \frac{\partial \tilde L}{\partial q_1}(q_1,q'_1)
\end{equation}
This is a wave equation, and since we have initial conditions~\eqref{wave init}, the solution~\eqref{wave sol} is valid.  (The general solution is not so obvious: if $\tilde L(x,y) = \tilde L(x,-y)$, then $\phi_1(s+t)$ and $\phi_2(s-t)$ are both solutions to this equation.  However it is not obvious to me how these two solutions should be combined, since this wave equation is non-linear.)

Next, the equations of motion for $(q_0,p_0)$ become
\begin{gather}
\label{xxx}
\dot p_0 = - \frac{\partial H_0}{\partial q_0}(q_0,p_0) - \frac{\partial \tilde L}{\partial q'_1(0)}(q_1(0),q'_1(0)) \\
\dot q_0 = \frac{\partial H_0}{\partial p_0}(q_0,p_0)
\end{gather}
where the last term of equation~\eqref{xxx} from the cross term from integrating by parts.  From equation~\eqref{wave sol} we see that $q'_1 = \dot q_1$.  Substituting this into equation~\eqref{xxx}, and noting that $q_1(0) = q_0$, we obtain the desired result.
\end{proof}

Finally, it is worth noting that under the initial conditions~\eqref{wave init} the Hamiltonian and Lagrangian evaluate to
\begin{gather}
H(q(t),p(t)) = H_0(q_0(t),p_0(t)) + \int_{T_0}^t \left(\frac{\partial \tilde L}{\partial \dot q_0}(q_0(\tau),\dot q_0(\tau))\right)\cdot \dot q_0(\tau) \, d\tau \\
L(q(t),\dot q(t)) = L_0(q_0(t),\dot q_0(t))
\end{gather}
In the case that $\tilde L(x,y) = \tfrac12 y \cdot B \cdot y$, the Hamiltonian becomes
\begin{equation}
H(q(t),p(t)) = H_0(q_0(t),p_0(t)) + \int_{T_0}^t \dot q_0(\tau) \cdot B \cdot \dot q_0(\tau) \, d\tau
\end{equation}

\section{The arrow of time}
\label{arrow of time}

Now that we have an explicit Hamiltonian that generates the equations of motion with friction, we should examine whether it satisfies the comments made at the end of Section~\ref{striving} on the ``arrow if time."

One way to reverse the effect of friction would be to first solve the equation backwards in time, from $t=T_1$ back to $t=T_0$, using initial conditions $q_1(s,T_1) = q_0(T_1)$, $p_1(s,T_1) = 0$ for ($s \le 0$).  Solve to find $q_1(s,T_0)$, noting that it will satisfy $q_1'(s,T_0) = - \dot q_1(s,T_0)$.  Then this is the initial data that is required to make the `wine glass come back together.'  Obviously not any initial data satisfying $q_1'(s,T_0) = - \dot q_1(s,T_0)$ will work, because the wave has to feed back into the equations of motion for $q_0$ exactly the energy it needs to satisfy the reverse effect of friction.

It would be interesting to see to what extent this equation is sensitive to the initial data, perhaps illustrating that it is very hard to get it exactly right to get the `wine glass to reconstruct.'  The idea would be to analyze the equations of motion for the full system on M, but either restricting to solutions in which $\dot q_1 = q'_1$ (in which case the solutions should be stable), or to the case in which $\dot q_1 = -q'_1$ (in which case the solutions should be unstable --- this latter case requires $\tilde L(x,y) = \tilde L(x,-y)$, otherwise the formulation of the problem would be harder).  Almost surely one will find that in the latter case, the equation is linearly unstable, corresponding to the fact that the Hessian of $\tilde L(x,\cdot)$ is positive definite.

\section{N\"other's Theorem}
\label{noether}

By applying N\"other's Theorem to our Lagrangian, we can easily obtain a version of N\"other's Theorem that applies to dissipative systems.  Suppose there is a flow on $M_0$ that preserves both $L_0$ and $\tilde L$, let us denote the parameter driving the flow by $\lambda$.  Then the following quantity is conserved.
\begin{equation}
\pi(t) = \frac{\partial L_1}{\partial \dot q_0}(q_0(t),\dot q_0(t)) \cdot \frac{\partial q_0(t)}{\partial \lambda}\Big|_{\lambda = 0} + \int_{T_0}^t \frac{\partial \tilde L}{\partial \dot q_0}(q_0(\tau),\dot q_0(\tau)) \cdot \frac{\partial q_0(\tau)}{\partial \lambda}\Big|_{\lambda = 0} \, d\tau
\end{equation}

Let us illustrate with a couple of examples of computing the momentum.  Suppose that $q_0$ is $n$ vectors in $\mathbb R^m$: $q_0 = (q_{01},q_{02},\dots,q_{0n})$.
Let $\mu_x$ be the momentum in the direction $x$
\begin{equation}
\mu(t) = \sum_i m_i \dot q_{0i} \cdot x
\end{equation}
which corresponds to the action
\begin{equation}
\label{action}
\lambda \mapsto (q_0 \mapsto (q_{0i} + \lambda x)_{i=1}^n))
\end{equation}
Suppose that
\begin{gather}
L_0(q_0,\dot q_0) = \sum_i \tfrac12 m_i|\dot q_{0i}|^2 + \sum_i \sum_j \tfrac12 v_{ij}(|q_{0i}-q_{0j}|) \\
\label{friction}
\tilde L(q_0,\dot q_0) = \eta \sum_i m_i \tfrac12 |\dot q_{0i}|^2
\end{gather}
where $v_{ij} = v_{ji}$.  This gives the equation of motion
\begin{equation}
m_i \ddot q_{0i} = - \sum_j \frac{(q_{0i}-q_{0j})}{|q_{0i}-q_{0j}|} v'_{ij}(|q_{0i}-q_{0j}|) - \eta m_i \dot q_{0i}
\end{equation}
Both of these Lagrangians are invariant under the action~\eqref{action}.  We see that
\begin{equation}
\pi(t) = \mu_x(t) + \eta \int_{T_0}^t \mu_x(\tau) \, d\tau
\end{equation}
is conserved, that is
\begin{equation}
\dot \mu_x + \eta \mu_x = 0 \qquad \Rightarrow \qquad \mu_x(t) = e^{-\eta t} \mu_x(0)
\end{equation}
Now suppose we change equation~\eqref{friction} to
\begin{equation}
\label{friction2}
\tilde L(q_0,\dot q_0) = \sum_i \sum_j \tfrac12 \eta_{ij}\left(\frac{\partial}{\partial t}|q_{0i}-q_{0j}|\right)
\end{equation}
where $\eta_{ij} = \eta_{ji}$, that is, all the damping effects are internal and depend completely on how the distances between the particles change.  The equation of motion is now
\begin{equation}
m_i \ddot q_{0i} = - \sum_j \frac{(q_{0i}-q_{0j})}{|q_{0i}-q_{0j}|} v'_{ij}(|q_{0i}-q_{0j}|) - \sum_j \frac{(q_{0i}-q_{0j})}{|q_{0i}-q_{0j}|} \eta'_{ij}\left(\frac\partial{\partial t}|q_{0i}-q_{0j}|\right)
\end{equation}
This Lagrangian is also invariant under the action~\eqref{action}.  We see that
\begin{equation}
\frac{\partial \tilde L}{\partial \dot q_0}(q_0(t),\dot q_0(t)) \cdot \frac{\partial q_0(t)}{\partial \lambda}\Big|_{\lambda = 0} = \sum_i \sum_j \frac{(q_{0i}-q_{0j})}{|q_{0i}-q_{0j}|} \eta'_{ij}\left(\frac\partial{\partial t}|q_{0i}-q_{0j}|\right) = 0
\end{equation}
Therefore $\pi = \mu_x$, and thus the momentum $\mu_x$ is conserved.

Let's try a similar argument with angular momentum.  Let $\alpha_A$ be the angular momentum about the anti-symmetric matrix $A$
\begin{equation}
\alpha(t) = \sum_i m_i \dot q_{0i} \cdot A \cdot q_{0i}
\end{equation}
which corresponds to the action
\begin{equation}
\label{action2}
\lambda \mapsto (q_0 \mapsto (e^{\lambda A} \cdot q_{0i})_{i=1}^n))
\end{equation}
Since $e^{\lambda A}$ is an orthogonal matrix, all of the above Lagrangians are invariant under the action~\eqref{action}.  For the case~\eqref{friction}
\begin{equation}
\pi(t) = \alpha_A(t) + \eta \int_{T_0}^t \alpha_A(\tau) \, d\tau
\end{equation}
is conserved, that is
\begin{equation}
\dot \alpha_A + \eta \alpha_A = 0 \qquad \Rightarrow \qquad \alpha_A(t) = e^{-\eta t} \alpha_A(0)
\end{equation}
For the case~\eqref{friction2}, we obtain
\begin{equation}
\frac{\partial \tilde L}{\partial \dot q_0}(q_0(t),\dot q_0(t)) \cdot \frac{\partial q_0(t)}{\partial \lambda}\Big|_{\lambda = 0}
= \sum_i \sum_j \frac{(q_{0i}-q_{0j})}{|q_{0i}-q_{0j}|} \eta'_{ij}\left(\frac\partial{\partial t}|q_{0i}-q_{0j}|\right) \cdot A \cdot q_{0i}
= 0
\end{equation}
and thus similarly $\alpha_A$ is conserved.

(We should add that technically Lagrangian~\eqref{friction2} doesn't quite fit into our framework, because it is not strictly convex in $\dot q_0$.  But this technicality is just that, and easily overcome.  One can either redo all the work, or add $\epsilon|\dot q_0|^2$ to this Lagrangian and take the limit as $\epsilon\to 0$.)

\section{An identity}
\label{identity}

In this section we present an identity which we believe might be useful for creating conserved quantities.  Suppose that equation~\eqref{wave} is satisfied.  Let
\begin{align}
E &= \tilde H(p_1,q_1) - \tilde L(q_1,q'_1) \\
F &= \dot q_1 \cdot \frac{\partial \tilde L}{\partial q_1'} (q,q')
\end{align}
Then
\begin{equation}
\frac{\partial E}{\partial t} = \frac{\partial F}{\partial s}
\end{equation}
In the case that $\tilde L(x,y) = y \cdot B(x) \cdot y$, so that $\tilde L(q_1,\dot q_1) = \tilde H(q_1,p_1)$, there is a symmetry in the formulas so that we also have
\begin{equation}
\frac{\partial E}{\partial s} = \frac{\partial F}{\partial t}
\end{equation}
and hence both $E$ and $F$ satisfy the standard wave equation:
\begin{equation}
\frac{\partial^2 E}{\partial t^2} = \frac{\partial^2 E}{\partial s^2}, \qquad
\frac{\partial^2 F}{\partial t^2} = \frac{\partial^2 F}{\partial s^2}
\end{equation}

\end{document}